\documentclass[a4paper,10pt]{llncs} 
\usepackage[utf8]{inputenc} 
\usepackage[english]{babel} 

\usepackage{hyperref} 
\usepackage{amsfonts}
\usepackage{amsmath}
\usepackage{enumerate}
\usepackage{array}
\usepackage{microtype}
\usepackage{enumitem}
\usepackage{pgf,tikz}
\usetikzlibrary{arrows,decorations.pathmorphing,calc,patterns,decorations.pathreplacing,fadings,shapes}

\usepackage{xspace}
\usepackage{graphicx}
\usepackage{ amssymb }
\usepackage{subfig}

\usepackage{nameref}
\makeatletter
\def\namedlabel#1#2{\begingroup
   \def\@currentlabel{#2}%
   \label{#1}\endgroup
}
\makeatother

\newcommand{\setU}{\ensuremath{{U}}\xspace}

\newcommand{\cost}{\ensuremath{{\rm{cost}}}\xspace}

\newcommand{\prob}{\ensuremath{L}\xspace}
\newcommand{\alp}{\ensuremath{\Sigma}\xspace}
\newcommand{\cli}{\textsc{Clique}\xspace}
\newcommand{\is}{\textsc{Independent Set}\xspace}
\newcommand{\setc}{\textsc{Set Cover}\xspace}
\newcommand{\lp}{\textsc{Longest Path}\xspace}
\newcommand{\np}{\ensuremath{\textsf{NP}}\xspace}
\newcommand{\tsp}{\text{TSP}\xspace}
\newcommand{\lmod}{\ensuremath{\mathrm{lm}}\xspace}
\newcommand{\verc}{\textsc{Vertex Cover}\xspace}
\newcommand{\covc}{\textsc{Connected Vertex Cover}\xspace}
\newcommand{\sptr}{\textsc{Internal Vertex SubTree}\xspace}
\newcommand{\outtree}{\textsc{Leaf Out Tree}\xspace}
\newcommand{\tw}{\textsc{Tree Width}\xspace}
\newcommand{\N}{\mathbb{N}}

\newcommand{\F}{\ensuremath{\mathcal{F}}\xspace}

\newcommand{\eg}{e.\,g.}
\newcommand{\ie}{i.\,e.}

\title{Reoptimization of Parameterized Problems}
\author{Hans-Joachim B\"{o}ckenhauer\inst{1} \and Elisabet
Burjons\inst{1} \and \\
Martin Raszyk\inst{1} \and Peter Rossmanith\inst{2}
\institute{Department of Computer Science, ETH Zurich, Switzerland\\\email{hjb@inf.ethz.ch, elisabet.burjons@inf.ethz.ch, m.raszyk@gmail.com}\and 
Department of Computer Science, RWTH Aachen, Germany\\\email{rossmani@informatik.rwth-aachen.de}}
}

\begin{document}

\maketitle

\begin{abstract}
 Parameterized complexity allows us to analyze the time complexity of problems with respect to a natural parameter
 depending on the problem. Reoptimization looks for solutions or approximations for problem instances when given solutions to neighboring instances.
 We combine both techniques, in order to better classify the complexity of problems in the parameterized setting.
 
 Specifically, we see that some problems in the class of compositional problems, which do not have polynomial kernels under standard
 complexity-theoretic assumptions, do have polynomial kernels under the reoptimization model for some local modifications.
 We also observe that, for some other local modifications, these same problems do not have polynomial kernels unless $\rm{NP}\subseteq \rm{coNP}/ \rm{poly}$.
 We find examples of compositional problems, whose reoptimization versions do not have polynomial kernels under any of the considered local modifications.
 Finally, in another negative result, we prove that the reoptimization version of \covc does not have a polynomial kernel unless
 \setc has a polynomial compression.
 
 In a different direction, looking at problems with polynomial kernels, we find that the reoptimization version of \verc has a polynomial kernel of size $2k$ using crown decompositions only,
 which improves the size of the kernel achievable with this technique in the classic problem.
\end{abstract}

\section{Introduction}

In this paper, we try to combine the techniques of reoptimization and parametrization in order to have a better understanding of what makes a problem
hard from a parameterized complexity point of view. The goal is, given a solution for an instance of a parameterized problem, try to 
look at local modifications and see if the problem becomes easier or if it stays in the same complexity class. For this, we look at
classical problems in parameterized complexity, whose complexity is well understood and classified. 

While the connections between reoptimization and parameterization were not systematically explored up to now, some links were already discovered. 
The technique of iterative compression which was introduced by Reed, Smith, and Vetta \cite{RSV04} was very successfully used to design parameterized algorithms, 
see the textbook by Cygan et al.~\cite{CFK+15} for an overview. It is closely related to common design techniques for reoptimization algorithms.
Abu-Khzam et al.~\cite{AEF+15} looked at the parameterized complexity of dynamic, reoptimization-related versions of dominating set and other problems, 
albeit more related to a slightly different model of reoptimization as introduced by Shachnai et al.~\cite{STT12}. Very recently, Alman, Mnich, 
and Williams \cite{AMW17} considered dynamic parameterized problems, which can be seen as a generalization of reoptimization problems. 

We start by introducing the main concepts of parameterized complexity and reoptimization that we are going to use in our results. 

\subsection{Parameterized Complexity}

Classical complexity theory classifies problems by the amount of time or space that is required by algorithms
solving them. Usually, the time or space in these problems is measured by the input size. However, 
measuring complexity only in terms of the input size ignores any structural information about the input 
instances, making problems appear sometimes more difficult than they actually are.

Parameterized complexity was developed by Downey and Fellows in a series of articles in the early 1990's~\cite{DF95,DF95-2}. 
Parameterized complexity theory provides a theory of intractability and of fixed-parameter tractability that relaxes
the classical notion of tractability, namely polynomial-time computability, by allowing non-polynomial computations only 
depending on a parameter independent of the instance size. 
For a deeper introduction to parameterized complexity we refer the reader to Downey et al.~\cite{DF13,FG06}.

We now introduce the formal framework for parameterized complexity that we use throughout the paper.
Let $\alp$ denote a finite alphabet and $\N$ the set of natural numbers.
A {\em decision problem} $\prob$ is a subset of $\alp^*$. 
We will call the strings $x\in\alp^*$, {\em input} of $\prob$, regardless of whether $x\in \prob$.
A {\em parameterized problem} is a subset $\prob\subseteq \alp^* \times \N$. An input $(x,k)$ to a parameterized 
language consists of two parts where the second part is the {\em parameter}. A parameterized problem \prob
is {\em fixed-parameter tractable} if there exists an algorithm that given an input $(x,k) \in \alp^*\times \N$,
decides whether $(x,k)\in\prob$ in $f(k)p(n)$ time, where $f$ is an arbitrarly computable function solely in $k$,
and $p$ is a polynomial in the total input length $n=|x|+k$. FPT is the class of parameterized problems which are 
fixed-parameter tractable.

A {\em kernelization} for a parameterized problem $\prob\subseteq \alp^* \times \N$ is an algorithm that, given 
$(x,k) \in \alp^*\times \N$, outputs in $p(n)$ time a pair $(x',k')\in\alp^* \times \N$, namely a {\em kernel}, such that 
$(x,k)\in \prob \iff (x',k')\in \prob$ and $|x'|,k'\le f(k)$, where $p$ is a polynomial and $f$ an arbitrary computable
function, $f$ is referred to as the {\em size} of the kernel.
If for a problem \prob, the size of the kernel $f$ is polynomial in $k$, we say that \prob has a {\em polynomial kernel}.
PK is the class of parameterized problems which have polynomial kernels.

A {\em Turing kernelization} is a procedure consisting of two parameterized problems $\prob_1$ and $\prob_2$ (typically $\prob_1=\prob_2$)
and a polynomial $g$ together with an oracle for $\prob_2$, such that, on an input $(x,k)\in \alp^*$, the procedures outputs the answer 
whether $x\in\prob_1$ in polynomial
time by querying the oracle for $\prob_2$ with questions of the form ``Is $(x_2,k_2)\in \prob_2$?'' for $|x_2|,k_2\le g(k)$.
Essentially, a Turing kernelization allows us to use an oracle for small instances,
in order to solve $\prob_1$ on a larger instance $(x,k)$. A {\em polynomial Turing kernelization} is a Turing kernelization 
where $g=(k)$ is a polynomial function. PTK is the class of parameterized problems which have polynomial Turing kernelizations.

The problem classes we defined up to now satisfy $\rm{PK}\subseteq\rm{PTK}\subseteq\rm{FPT}$.
There are well known problems, however, that are not known to be FPT. For example, $k$-\cli, which is the problem of identifying whether
a graph $G$ contains a clique of size $k$, is not contained in FPT under some standard complexity-theoretic assumptions. 
Neither the complementary problem $k$-\is, which is the problem of identifying whether 
a graph $G$ contains an independent set of size $k$, or the $k$-\setc problem, where given a universe set \setU and a family \F of subsets of \setU, we are asked
to determine whether there is a subset of \F of size $k$ which contains every element of \setU.
For these problems outside FPT there is a further
classification of their hardness in terms of the so-called {\em W hierarchy} consisting of classes $W[t]$ for $t\in\N$, such that
$W[t]\subseteq W[t+1]$. Moreover, $\rm{FPT}\subseteq W[1]$. 
For the definition of these classes and the theory behind it see~\cite{DF13}. 
In this paper, we will only use the classes $W[1]$ and $W[2]$. For them we have the following characterizations in terms of complete problems:
$k$-\cli and $k$-\is are complete for $W[1]$, and $k$-\setc is complete for $W[2]$.

\subsection{Reoptimization}

Often, one has to solve multiple instances of one optimization problem
which might be somehow related. Consider the example of a timetable for
some railway network. Assume that we have spent a lot of effort and
resources to compute an optimal or near-optimal timetable satisfying all
given requirements. Now, a small local change occurs like, \eg, the closing
of a station due to construction work. This leads to a new instance of our
timetable problem that is closely related to the old one. Such a situation
naturally raises the question whether it is necessary to compute a new
solution from scratch or whether the known old solution can be of any help.
The framework of \emph{reoptimization} tries to address this question: We
are given an optimal or nearly optimal solution to some instance of a hard
optimization problem, then a small local change is applied to the instance,
and we ask whether we can use the knowledge of the old solution to
facilitate computing a reasonable solution for the locally modified
instance. It turns out that, for different problems and different kinds of
local modifications, the answer to this question might be completely
different. Generally speaking, we should not expect that solving the
problem on the modified instance optimally can be done in polynomial time,
but, in some cases, the approximability might improve a lot.

This notion of reoptimization was mentioned for the first time by
Sch\"affter \cite{Sch97} in the context of a scheduling problem. Archetti
et al.\ \cite{ABS03} used it for designing an approximation algorithm for
the metric traveling salesman problem ($\Delta$\tsp) with an improved running time,
but still the same approximation ratio as for the original problem. But the
real power of the reoptimization concept lies in its potential to improve
the approximation ratio compared to the original problem. This was observed
for the first time by B\"ockenhauer et al.\ \cite{BFH+06} for
the $\Delta$\tsp, considering the change of one edge weight as a local
modification. Independently at the same time, Ausiello et al.\
\cite{AEM+06} proved similar results for TSP reoptimization under the local
modification of adding or removing vertices.

Intuitively, the additional information that is given in a reoptimization
setup seems to be rather powerful.   Intriguingly, many reoptimization
variants of \np-hard optimization problems are also \np-hard. A general
approach towards proving the \np-hardness of reoptimization problems uses a
sequence of reductions and can on a high level be described as follows
\cite{BHM+08}: Consider an \np-hard optimization problem \prob, a local
modification \lmod, and a resulting reoptimization problem \lmod-\prob.
Moreover, suppose we are able to transform an efficiently solvable instance
$x'$ of \prob to any instance $x$ of \prob in a polynomial number of
local modifications of type \lmod.  Then, any efficient algorithm for
\lmod-\prob could be used to efficiently solve \prob, thus the NP-hardness of \prob
implies the hardness of \lmod-\prob.

\subsection{Reoptimization of Parameterized Problems}

Now that we have seen the main concepts of parameterized complexity and reoptimization, 
we will formally define an instance for a reoptimization parameterized problem \lmod-\prob.

Given a parameterized problem $\prob$, a \emph{solution} $s$ for a problem instance $(x,k)$ is a 
witness of size $|s|\le p(|x|)$ for some polynomial $p$, with which we can check in polynomial time
that $(x,k)\in \prob$.
In order to measure how good a solution is, we have to define the \emph{cost} of the solution.
For some parameterized problems, the parameter is already a measure of the goodness of the solution. For these problems,
we will say a solution $s$ has \emph{cost} $k$ if $(x,k)\in \prob$ but $(x,k')\not\in \prob$ 
for any $k'<k$, if \prob is a minimization problem, and $k'>k$, if \prob is a maximization problem.

In problems where the parameter $k$ is an intrinsic value of the instance rather than a quality measure,
we have to define an extra parameter $\gamma$ measuring the quality of the solutions.
A cost function $\cost(\cdot)$, is a polynomially computable function that, given a solution $s$ to an 
instance $(x,k)$ computes the value of $\gamma$ corresponding to this solution.
Often, this parameter will be the size of the solution, but other parameters can be used.
In these problems, we redefine an instance to be a triple $(x,k,\gamma)$ where $(x,k,\gamma)\in \prob$ if and only if $(x,k)\in \prob$
and there exists a solution $s$ with $\cost(s)\le \gamma$ if \prob is a minimization problem and $\cost(s) \ge \gamma$ 
if \prob is a maximization problem. 

From now on, we assume that $k$ is a cost parameter unless otherwise specified, and thus, we refer to instances as pairs $(x,k)$.
An \emph{instance} of a reoptimization problem \lmod-\prob consists of: an instance
of the parameterized problem $\prob$, $(x,k)\in \alp^*\times \N$,
together with a solution $s$ with $\cost(s)\le k$ for 
minimization problems and $\cost(s)\ge k$ for maximization problems if it exists, \ie, if $(x,k)\in \prob$, or $\bot$ if $(x,k)\notin \prob$, and a 
locally modified instance $(x_{\lmod},k')$, where $k'\le k$ for minimization problems and $k'\ge k$ for maximization problems and $k'\in f(k)$ where $f$ is a computable function. 
We say that $((x,k),s,(x_{\lmod},k'))\in \text{\lmod-\prob}$ if and only if $(x_{\lmod},k')\in \prob$.

We will also define a polynomial kernel for a reoptimization instance \\
$((x,k),s,(x_{\lmod},k'))$ as a polynomial kernel for 
$(x_{\lmod},k')$. This makes sense because $((x,k),s,(x_{\lmod},k'))\in \text{\lmod-\prob}$ if and only if $(x_{\lmod},k')\in \prob$.

\subsection{Our Contribution}

In this paper, we use reoptimization techniques to solve
parameterized problems or to compute better kernels for them. In particular, we show in Section~\ref{sec:comp} that some
compositional parameterized problems~\cite{BDFH09}, which do not have polynomial kernels under
standard complexity-theoretic assumptions, do have polynomial kernels in a reoptimization setting, for some local modifications. 
Moreover, in Section~\ref{sec:recowopk}, we show that, under the opposite local modifications, those same problems do not have
polynomial kernels unless $\rm{NP}\subseteq\rm{coNP}/\rm{poly}$. We also show that, some compositional problems do not have polynomial
kernels under any of the standard local modifications for graph problems, \ie, vertex or edge addition or deletion.

Section~\ref{sec:cvc} contains a reduction of \setc parameterized by the size of the universe to \covc, that shows  
that the reoptimization of \covc under edge addition does not have a polynomial kernel unless $\rm{NP}\subseteq\rm{coNP}/\rm{poly}$.

We then show in Section~\ref{sec:vc} that, for the reoptimization version of the vertex cover problem, 
the crown decomposition technique yields a kernel of size $2k$.

\section{Kernels for Compositional Problems}
\label{sec:comp}

Bodlaender et al.~\cite{BDFH09} define the concept of \emph{compositional parameterized problems}, specifically OR-compositional and AND-compositional problems, for both of which no polynomial kernel exists 
under standard complexity-theoretic assumptions. In this section, we see that some of these problems do indeed have polynomial kernels in a reoptimization setting, 
where an optimal solution or a polynomial kernel is given for a locally modified instance.

\subsection{Preliminaries}

A characterization of OR-compositional graph problems is the following.

\begin{definition}[\cite{BDFH09}]\label{def:compositional}
 Let $\prob$ be a parameterized graph problem. If for any pair of 
 graphs $G_1$ and $G_2$, and any integer $k\in\N$, we have 
 \[(G_1,k)\in \prob \vee (G_2,k)\in \prob \iff (G_1\cup G_2,k)\in \prob,\]
 where $G_1\cup G_2$ is the disjoint union of $G_1$ and 
 $G_2$, then $\prob$ is \emph{OR-compositional}. 
\end{definition}

Now, if we define the complement of a problem, an analogous characterization can be defined that will
identify problems whose complement is OR-compositional, the so-called AND-compositional problems.

\begin{definition}
 Let $\prob$ be a parameterized decision problem. The \emph{complement $\bar{\prob}$ of $\prob$}, is the decision problem resulting from reverting the yes- and no-answers.
\end{definition}

\begin{definition}[\cite{BDFH09}]
 \label{def:compositionalcompl}
 Let $\prob$ be a parameterized graph problem. If for any pair of 
 graphs $G_1$ and $G_2$, and any integer $k\in\N$, we have 
 \[(G_1,k)\in \prob \wedge (G_2,k)\in \prob \iff (G_1\cup G_2,k)\in \prob,\]
 where $G_1\cup G_2$ is the disjoint union of $G_1$ and 
 $G_2$, then $\bar{\prob}$ is OR-compositional and $\prob$ is \emph{AND-compositional}. 
\end{definition}

Bodlaender et al.~\cite{BDFH09} showed the following result.

\begin{theorem}[\cite{BDFH09}]
\label{thm:compositional}
NP-hard OR-compositional problems do not have polynomial kernels, unless $\rm{NP}\subseteq\rm{coNP}/\rm{poly}$, \ie, the polynomial hierarchy collapses.
\end{theorem}

Moerover, Drucker~\cite{D12} was able to show the following. 

\begin{theorem}[\cite{BDFH09}]\label{thm:compositionalcompl}
Unless $\rm{NP}\subseteq\rm{coNP}/\rm{poly}$, NP-hard AND-compositional problems
do not have polynomial kernels.
\end{theorem}

We prove in this section that reoptimization versions of some OR-compositional or AND-compositional problems have polynomial kernels. Let us see now which local modifications
will provide these results.

When we talk about graph problems in a reoptimization setting, four local modifications come to mind immediately, namely edge addition and deletion, and vertex addition and deletion.
We now define them formally.

Given a graph $G=(V,E)$, and a pair of non-neighboring vertices $u,v\in V$, we denote an \emph{edge addition} $(V,E\cup\{u,v\})$ by $G+\{u,v\}$, or $G+e$ where $e=\{u,v\}$.
Analogously, for \emph{edge deletion}, given an edge $e\in E$, $G-e$ is the graph $(V,E-\{ e\})$.
Furthermore, for \emph{vertex deletion}, given a vertex $v\in V$, $G-v$ is the subgaph induced by $V-\{v\}$, \ie, $(V-\{ v\}, E')$ where $E'$ is $E$ without the edges incident to $v$. 
Finally, in the case of \emph{vertex addition}, given a new vertex $v$ and a set of edges $E'\subseteq\bigcup_{u\in V}\{u,v\}$, $G+v$ is $G=(V\cup \{v\}, E\cup E')$.
Given a graph problem $\prob$, we call the reoptimization version of $\prob$ under edge addition, edge deletion, vertex addition, and vertex deletion $e^+ $-$\prob$, $e^- $-$\prob$, $v^+ $-$\prob$, and $v^- $-$\prob$, respectively.

We now give an example of a OR-compositional FPT problem that is in PK under reoptimization conditions. We want to see which are the conditions that allow us to find a kernel in this setting.

\subsection{Internal Vertex Subtree}\label{subsec:sptr}

A \emph{subtree} $T$ of a graph $G$ is a (not necessarily induced) subgraph of $G$ which is also a tree. The vertices of a tree can be classified into two categories: \emph{leaves} are vertices of degree 1, and 
\emph{internal} vertices are vertices of higher degree.
Let us consider the following parameterized decision problem called the \sptr problem.
Given a graph $G$ and an integer $k$, we have to determine whether $G$ contains a subtree with at least $k$ internal vertices.

The connected version of this problem, where we consider as input only pairs $(G,k)$ where $G$ is connected, is called \textsc{Maximum Internal Spanning Tree} and has a polynomial kernel of size $3k$ using the crown lemma~\cite{FGST13}, 
and an improved polynomial kernel of size $2k$~\cite{LWCC15}.
However, the general version of this problem does not have a polynomial kernel, unless $\rm{NP}\subseteq\rm{coNP}/\rm{poly}$. 
Let us see this.

\begin{theorem}
 \sptr in general graphs does not have a polynomial kernel unless $\rm{NP}\subseteq\rm{coNP}/\rm{poly}$.
\end{theorem} 

\begin{proof}
 Observe first that \sptr is OR-compositional. As required by Definition~\ref{def:compositional}, given two connected graphs $G_1$ and $G_2$, if one of them has a subtree with
 $k$ internal vertices, then the disjoint union of them, \ie, the graph with two connected components $G_1$ and $G_2$ will also have one, the same one that was in 
 $G_1$ or $G_2$. This argument easily extends to arbitrary graphs $G_1$ and $G_2$. As for the reverse implication, if a graph $G$ contains two connected components $G=G_1\cup G_2$ and 
 has such a subtree, then the whole subtree, which is connected, must be contained in one of the components, meaning that either $(G_1,k)\in \sptr$ or  $(G_2,k)\in \sptr$.
 
 Moreover, \sptr is NP-complete (in particular NP-hard). This is because there is a straightforward reduction from Hamiltonian path (see~\cite{OY11}), which is well known to be NP-complete.
 
 Finally, we see that, by Theorem~\ref{thm:compositional}, \sptr does not have a polynomial kernel unless $\rm{NP}\subseteq\rm{coNP}/\rm{poly}$.
\qed \end{proof}

Now, we are going to prove that $e^+$-\sptr has a polynomial kernel.

\begin{theorem}
 $e^+$-\sptr has a polynomial kernel of size $2k$.
\end{theorem}

\begin{proof}
 Let us consider an instance $((G,k),T, (G+e,k))$ for $e^+$-\sptr. Recall that $T$ is a subtree of $G$ with at least $k$ internal vertices (\ie, a solution for $(G,k)$) if it exists
 or a trivial tree $T$ (with less than $k$ internal vertices) otherwise.
 The following procedure gives a kernel of size $2k$ for the modified input $(G+e,k)$.
 
 If $T$ is a subtree with at least $k$ internal vertices, then $T$ is also a valid solution for $G+e$, thus any instance where $(G,k)\in \sptr$ implies immediately that $((G,k),T, (G+e,k))\in e^+-\sptr$,
 thus any trivial instance $(H,k)\in \sptr$ of size $\le 2k$ is a kernel for $e^+$-\sptr.
 
 On the other hand, if $(G,k)$ contains no such tree, then it suffices to check for $(G+e,k)$, whether the connected component containing the edge $e$
 has such a subtree. 
 Because any other connected component of $G+e$ is identical to a component in $G$, and we know that those components do not 
 contain any subtree with at least $k$ internal vertices.
 This means that $(G+e,k)\in\sptr$ if and only if the connected component containing $e$ has a subtree with at least $k$ internal vertices. And thus, a kernel for this component is equivalent to a kernel of the
 whole instance.
 As we know that a $2k$ kernel exists for the connected case~\cite{LWCC15}, we can obtain one such kernel for the connected component containing $e$, thus we have provided a kernel of size $2k$ for
 $((G,k),T,(G+e,k))$.
\qed \end{proof}

This shows that $e^+$-\sptr is in PK.
Observe that, in this case, we would be able to find a kernel for the modified instance by using the same procedure, even if we were given only a Yes/No answer or a polynomial kernel
instead of a solution for the non-modified instance.
This is because given an instance $(G,k)$ for \sptr, if  we are guaranteed this instance has a  subtree with $k$ internal vertices, then for sure $(G+e,k)$ also has one, on the other hand,
if we are guaranteed that $(G,k)$ does not have such a subtree, then if one should exist for $(G+e,k)$, it would be found in the component that contains $e$, and thus we could build a kernel for 
that component. In the case we are given just an instance $(G,k)$ and a polynomial kernel for this instance, the way to build a kernel for $(G+e,k)$ is just to build a kernel for the component that contains
$e$, and give as polynomial kernel for $(G+e,k)$ the kernel obtained by taking a disjoint union of both kernels. 
We can find through the first kernel 
if $(G,k)$ has a 
subtree with $k$ internal vertices and in this case second kernel is not relevant, otherwise, we can look at the second kernel to determine if the component containing
the edge $e$ has a spanning tree with $k$ internal vertices, thus solving the instance $(G+e,k)$.

\subsection{Generalization}

To begin with, we observe that, in order for a problem to be solvable in an analogous way to the problem above, it is important that the property defining the problem is maintained under the local modification considered.
For instance, a subtree of a graph $G$ is also a subtree of the same graph with an added edge, $G+e$. However, the same does not hold for edge deletion, because the deleted edge might be part of 
the chosen subtree for $G$. In order to formalize this, we define the following:

\begin{definition}[Monotone Graph Problem]
 A graph problem $\prob$ is called \emph{monotone} if it is closed under removal of edges and vertices. 
 That is, if an instance $(G,k)\in\prob$, then $(G-e,k)\in \prob$ for every $e\in E$ and $(G-v,k)\in \prob$ for every $v\in V$.
\end{definition}

\begin{definition}[Comonotone Graph Problem]
 Similarly, a graph problem $\prob$ is called \emph{comonotone} if it is closed under addition of edges and vertices.
\end{definition}

We see in this subsection how to construct polynomial kernels for the
reoptimization versions of some compositional graph problems that are monotone or comonotone and that are not in PK.

We realize that, in order to get similar results as in the examples above, we need the following conditions. Let $\prob$ be a graph problem;
\begin{enumerate}
 \item $\prob$ is compositional and NP-hard. 
 \item Any instance of the parameterized problem $\prob$ has a polynomial kernel on the connected component of a given vertex or edge, or
 an instance of the rooted version of the problem $\prob^*$ has a polynomial kernel (informally, in a rooted version of a problem,
 any given instance contains a distinguished vertex, the root, and the solution must contain
 this vertex).
 \item The problem is monotone or comonotone.
\end{enumerate}

The first condition ensures that the considered problem does not have a polynomial kernel unless $\rm{NP}\subseteq\rm{coNP}/\rm{poly}$,
which makes results on reoptimization interesting. 
The second condition allows us to find kernels locally.
The third condition allows the modification to only affect the solution locally, whereas other modifications
could potentially require to look at the whole instance for a solution. Let us formalize this.

We define the \emph{environment} of an edge $e$ or a vertex $v$ in a graph $G+e$ or $G+v$,
as the connected component that contains $e$ or $v$.
For edge and vertex deletions, we say that the environment of $e$ or $v$ in a graph $G-e$ or $G-v$ are the connected components that
are modified or generated when $e$ or $v$ is deleted from $G$.

Given an instance $(G,k)$ for a parameterized problem $\prob$ on graphs whose solution can be described by a subset of
vertices, an instance of the rooted version $\prob^*$ of the problem 
is a triplet
$(G,v,k)$ where $(G,v,k)\in \prob^*$ if and only if $(G,k)\in \prob$ and there exists a solution containing $v$.
We will say that a kernel $(G',k')$ for an instance $(G,k)$ is a \emph{$v$-rooted} kernel if it is a kernel for $(G,v,k)$ for the rooted problem $\prob^*$,
\ie, such that $(G',k')\in \prob$ if and only if $(G,k)\in \prob$ and has a solution (represented by a subset of vertices and edges) that contains $v$.

Finally, given an instance $(G,k)\in \prob$ for a graph problem \prob whose solution can be described by a subset of vertices, we say that the solution
$S\subseteq V$ to $(G,k)$ is a \emph{witness solution} if the subgraph $H$ induced by $S$ is an instance for \prob and $(H,k)\in \prob$, and moreover
$S$ is a solution for any supergraph $G'$ for which $H$ is an induced subgraph.
Essentially, we require all of the vertices which are necessary for the solution to be valid, to be part of the solution subset, and we require that the solution keeps being valid
for any supergraph. This last requirement is automatically satisfied in comonotone graph problems, which would allow us to relax the definition, but is needed in the case of monotone
graph problems.

We are now ready to formally state the theorems, which generalize the results we have for \sptr.

\begin{theorem}\label{thm:orcompmon}
 Let $\prob$ be a parameterized NP-hard OR-compositional monotone graph problem.
 If, for every instance $(G,k)$, we can compute a polynomial kernel for an environment of any edge $e\in E$
 or if there exist witness solutions for instances in \prob and we can compute a rooted polynomial kernel for any vertex $v\in V$, then $e^-$-$\prob$ is in PK.
 Moreover, if for every instance $(\hat{G},\hat{k})$, we can compute a polynomial kernel for an environment of any vertex $v\in \hat{V}$,
 then $v^-$-$\prob$ is in PK.
\end{theorem}

In the same way, we can state a similar theorem for comonotone graph problems with the complementary reoptimization steps, 
namely edge and vertex addition.

\begin{theorem}\label{thm:orcompcomon}
 Let $\prob$ be a parameterized NP-hard OR-compositional comonotone graph problem.
 If, for every instance $(G,k)$, we can compute a polynomial kernel for an environment of any edge $e\in E$
 or if there exist witness solutions for instances in \prob and we can compute a rooted polynomial kernel for any vertex $v\in V$, then $e^+$-$\prob$ is in PK.
 Moreover, if for every instance $(\hat{G},\hat{k})$, we can compute a polynomial kernel for an environment of any vertex $v\in \hat{V}$
 or a polynomial kernel rooted to $v$ for any $v\in \hat{V}$, then $v^+$-$\prob$ is in PK.
\end{theorem}


We state a proof for Theorem~\ref{thm:orcompmon} in the case of edge deletion and the rest of the cases will be proven by analogy to it.

\begin{proof}[of Theorems~\ref{thm:orcompmon} and \ref{thm:orcompcomon}]
 If an instance $(G,k)$ for a monotone parameterized graph problem \prob is a yes instance then, we can construct a trivial yes-kernel for $(G-e,k)$. 
 
 Otherwise, $(G,k)\notin \prob$. 
 It is important to observe, that in case of an edge deletion, given an instance $G$ and an edge $e$, the environment of $e$ might contain two connected components.
 
 If the considered problem $\prob$ has polynomial kernels for an environment of an edge, it will have a polynomial kernel for
 the graph $G-e$ because, if we take one polynomial kernel for each component adjacent to $e$,
 the size of the union of these kernels
 is still polynomial in $k$.
 Moreover, none of the other components
 are modified, thus any solution found for $G-e$ must be in one of the newly generated components. Thus making the generated kernel, a valid kernel for $(G-e,k)$.
 
 On the other hand, if for any instance of the considered problems we can compute a kernel for any vertex $v \in V$, we argue that a kernel rooted in a vertex adjacent to 
 $e$ is a kernel for $(G-e, k)$ in this case.
 Let us assume that $(G-e,k)$ has a witness solution $S'$ that does not contain the vertices adjacent to $e$, this would mean that the subgraph $H$ induced by $S'$ has $G$ as a supergraph.
 By the definition of witness solution, if $G$ is a supergraph of $H$, then $(G,k)\in \prob$, thus contradicting the assumption that  $(G,k)\notin \prob$. 
 
 The cases for vertex deletion and edge and vertex addition in comonotone graph problems are completely analogous, which proves
 Theorems~\ref{thm:orcompmon} and \ref{thm:orcompcomon}.\qed
 \end{proof}
 
 In the case of vertex deletion, the number of newly generated connected components can be as high as the degree of the deleted vertex $v$. It is important to point out
 that, in this case, in order to have a polynomial kernel for an environment, it might not be enough that $\prob$ 
 is in PK if restricted to connected instances (which was true for the previous cases).
 This is also the reason why, if a monotone problem has a rooted kernel, the theorem still does not hold up for $v^-$-$\prob$, as we would need to make sure that the deleted vertex has restricted degree, too.
 For any vertex with degree superpolynomial in $k$, even the existence of rooted kernels for all of its neighbors would not provide us with a polynomial kernel for 
 $G-v$.

In general, polynomial kernels cannot be built for OR-compositional hard problems because one might have a lot of connected components and one can only build a polynomial kernel for each connected component.
In fact, in the next section, we will see that in general some of the reoptimization versions of OR(or AND)-compositional problems do not have polynomial kernels.

If we now come back to the \sptr problem that we saw in Subsection~\ref{subsec:sptr}, we realize that not only the conditions are satisfied to apply Theorem~\ref{thm:orcompcomon} for an edge addition, 
but also to apply it to a vertex addition. Thus, we conclude that $v^+$-\sptr is in PK.

A problem where we can find a rooted kernel is \outtree, known sometimes in the literature as $k$-\outtree. Given a directed graph $D$ and an integer $k$, we are asked to compute a tree in $D$ with at least $k$ leaves. 
We first observe that this problem is comonotone.
Moreover, the rooted version of this problem is in PK with a quadratic kernel as was seen by Daligault and Thomassé~\cite{DT09}, furthermore
\outtree has no polynomial kernel unless $\rm{coNP}\subseteq\rm{NP}/\rm{poly}$ as pointed out by Fernau et al.~\cite{FFLRSV09}, due to the fact that the problem is OR-compositional and NP-hard. 
Then, applying Theorem~\ref{thm:orcompcomon}, we deduce that $e^+$-\outtree and $v^+$-\outtree are in PK.

Yet another problem that falls into this category is \cli on graphs of maximum degree $d$ ($d$-\cli), this problem has a polynomial kernel for the rooted case. It is OR-compositional and does not have a polynomial kernel in the general case
as observed by Hermelin et al.~\cite{HK13}, thus we deduce that we can apply Theorem~\ref{thm:orcompcomon} and thus $e^+$-$d$-\cli and $v^+$-$d$-\cli are in PK.


To sum it up we have the following corollary:

\begin{corollary}\label{cor:orexamples}
 The following problems are in PK:
 \begin{itemize}
  \item $v^+$-\sptr
  \item $e^+$-\outtree and $v^+$-\outtree
  \item $e^+$-$d$-\cli and $v^+$-$d$-\cli
 \end{itemize}

\end{corollary}

Now, we can think about the complementary version of the problems described, where a property is required in every component in order for a solution to exist, \ie, AND-compositional problems.
Again, we can state a pair theorems for AND-compositional problems analogous to Theorems~\ref{thm:orcompmon} and \ref{thm:orcompcomon} based on the considered local modifications.

\begin{theorem}\label{thm:andcompmon}
 Let $\prob$ be a parameterized NP-hard AND-compositional monotone graph problem.
 If, for every instance $(G,k)$, we can compute a polynomial kernel for an environment of any edge $e\in E$, 
 then $e^+$-$\prob$ is in PK.
 Moreover, if for every instance $(\hat{G},\hat{k})$, we can compute a polynomial kernel for an environment of any vertex $v\in \hat{V}$,
then $v^+$-$\prob$ is in PK.
\end{theorem}

\begin{theorem}\label{thm:andcompcomon}
 Let $\prob$ be a parameterized NP-hard AND-compositional comonotone graph problem.
 If, for every instance $(G,k)$, we can compute a polynomial kernel for an environment of any edge $e\in E$, 
then $e^-$-$\prob$  is in PK.
 Moreover, if for every instance $(\hat{G},\hat{k})$, we can compute a polynomial kernel for an environment of any vertex $v\in \hat{V}$, 
 then $v^-$-$\prob$  is in PK.
\end{theorem}

The proofs of these theorems are analogous to the ones for Theorems \ref{thm:orcompmon} and \ref{thm:orcompcomon}. We again prove one of the statements
and the rest are proven analogously.

\begin{proof}
 Given a solution to an instance $(G,k)$ for a monotone problem $\prob$, we find a polynomial kernel for an instance $(G+e,k)$ as follows.
  
 If $(G,k)\notin \prob$, then for sure $(G+e,k)$ does not have a solution. If $(G+e,k)$ was in $\prob$, then
 any solution for $(G+e,k)$ would also be a solution for $(G,k)$ because $\prob$ is monotone.
 
 If $(G,k)$ has a solution $S$, 
 $S$ might not be a valid solution for $G+e$. However, because $\prob$ is AND-compositional, it means that
 the required property is already satisfied in every component of $G+e$ except, maybe, in the environment of $e$.
 This means that checking whether the environment of $e$ is in
 $\prob$ is enough to ensure that $G-e\in \prob$. Thus, a kernel for an environment of $e$ will be a kernel for the reoptimization instance.
%
 
 The cases for edge and vertex addition in monotone graph problems and vertex deletion in complement of monotone graph problems are completely analogous.
 
 In the case of vertex deletion, we have to, again, take into account that the number of newly generated components might make it impossible to
 find a local polynomial kernel.
\qed \end{proof}

Observe, that the theorems for AND-compositional problems do not mention local kernels for rooted versions of the problem. This is because, when constructing
a kernel for the reoptimization version of an OR-compositional problem, we are given an instance without a solution, and then the modified instance might have a solution.
Intuitively, it is clear that the new solution has to be around the local modification.
In AND-compositional problems, however, the proceadure is exactly the opposite. Given an instance that has a solution, we are provided with a local modification that renders
that solution useless. Essentially, we need to make sure, that the component or components affected by the modification still have a solution. This solution will thus, not 
need to be a new solution, but one that might already have existed within the component in the original instance, but that was not given in the reoptimization instance, as the reoptimization instance
only requires one solution for the original instance to be given.

%
%
%
%
%
%

\section{Reoptimization Compositional Problems without Polynomial Kernels}\label{sec:recowopk}

We have just presented a general strategy to construct polynomial kernels for reoptimization versions of OR-compositional and AND-compositional problems.
Let us focus now on proving which of these problems do not have polynomial kernels even under reoptimization conditions. That is, problems where even
knowing an optimal solution for a neighboring instance does not help to build a kernel for the given instance.

First, we give an intuitive approach to the kernelization results for compositional problems.
In order to build kernels for reoptimization versions of OR-compositional and AND-compositional problems,
we took a local modification that would not break the solution, i.e., a local modification that would respect
the monotonicity properties of the problem. Through this monotonicity we then could build a kernel centered on 
the local modification, knowing that the rest of the solution remains valid.

Now we try to do the opposite. That is, we will take local modifications which go against the monotonicity
of the problem properties. Then, we use a clever built-in solution for a neighboring instance that will be broken when the local
modification occurs, yielding the knowledge of the neighboring solution useless. Let us observe this through an example.

\subsection{$e^-$-Longest Path}

We now give an example of an OR-compositional FPT problem that is not in PK under reoptimization conditions. We want to see which are the conditions
that make finding a kernel in this setting as difficult as the original problem.

In the parameterized \lp problem, the goal is, given an instance $(G,k)$ to determine whether $G$ contains a path of length at least $k$.

It is easy to see that this problem is OR-compositional, and it is NP-complete~\cite{CLRS09}, so in general, according to Theorem~\ref{thm:compositional}, 
it is not in PK unless $\rm{NP}\subseteq\rm{coNP}/\rm{poly}$.

We are going to show now that this even holds for certain reoptimization variants.

\begin{theorem}\label{thm:longestpath}
 $e^-$-\lp and $v^-$-\lp do not have polynomial kernels unless $\rm{NP}\subseteq\rm{coNP}/\rm{poly}$.
\end{theorem}

\begin{proof}
 We prove the claim by providing a reduction from \lp to $e^-$-\lp.
 
 Given an instance $(G,k)$ for \lp, we construct an instance for $e^-$-\lp as follows.
 Given $P_k$ a path of length $k$, let $((G\cup P_k,k), P_k, (G\cup (P_k-e),k))$  be an instance
 for $e^-$-\lp  where $e$ is an edge in $P_k$. 
 We observe, that after deleting an edge from $P_k$, $P_k$ is no longer a path of length $k$ and thus 
 $((G\cup P_k,k), P_k, (G\cup (P_k-e),k))\in e^-$-\lp  if and only if $(G,k)\in \lp$. Moreover, the solution
 $P_k$ does not provide any information about the graph $G$ in which the new solution must be found.
 Thus, if $e^-$-\lp would belong to PK, given any instance $(G,k)$ of \lp, we would be able to construct a kernel for 
 it by providing a kernel for $((G\cup P_k,k), P_k, (G\cup (P_k-e),k))$. But \lp is not in PK
 unless $\rm{NP}\subseteq\rm{coNP}/\rm{poly}$, thus proving the statement.
 
 The reduction for $v^-$-\lp is completely analogous.
\qed \end{proof}

The insight that this example provides is that, if a reoptimization instance has an easy-to-spot solution that is
not available after the reoptimization step, then solving this instance might be as hard as solving the problem in general
without any extra information.

\subsection{General Results}

In order to prove a general result about reoptimization versions of OR- and AND-compositional problems we need to understand
what an easy solution, or an easy-to-break solution looks like.

We say that a graph $G$ is {\em extremal} with respect to the problem $\prob$ and the parameter $k$ if it is a maximal
graph such that $(G,k)\in \prob$, \ie, $(G+e,k)\not\in \prob$ 
for any edge, or if it is a minimal graph such that $(G,k)\in \prob$, \ie,  $(G-e,k)\not\in \prob$ and $(G-v,k)\not\in \prob$ for any edge or vertex.

Extremal graphs, if easy to construct, will help us design, from an instance for a graph problem \prob, an instance for its reoptimization version such that
the existence of polynomial kernels for the reoptimization version would imply that \prob is also in PK.

\begin{theorem}\label{thm:orcompneg}
 Let \prob be a monotone (comonotone) NP-hard OR-compositional graph problem. If, given an instance $(G,k)$ for \prob, we can compute
 in time polynomial in $k$ an extremal graph with respect to $k$. Then,
 $e^+$-$\prob$ ($e^-$-$\prob$ and $v^-$-$\prob$ resp.) are not in PK unless $\rm{NP}\subseteq\rm{coNP}/\rm{poly}$.
\end{theorem}

\begin{proof}
 Let \prob be a monotone NP hard OR-compositional graph problem and let $(G,k)$ be an instance for \prob.
 Let then $H$ be an extremal graph with respect to \prob and $k$.
 
 The instance $((G\cup H,k), H, (G\cup (H+e),k))$ for $e^+$-$\prob$ will only be in $e^+$-$\prob$ if
 $(G,k)\in \prob$, as $(H+e,k)\not\in \prob$ by construction.
 
 Thus, if $e^+$-\prob would be in PK, given any instance $(G,k)$ of \prob we would be able to construct a kernel for 
 it by providing a kernel for $((G\cup H,k), H, (G\cup (H+e),k))$. But \prob is not in PK
 unless $\rm{NP}\subseteq\rm{coNP}/\rm{poly}$, thus proving the statement.
 
 An analogous construction proves the statement for problems that are comonotone.
 \qed
\end{proof}

In particular, as a corollary we have:

\begin{corollary}\label{cor:negex}
  The following problems are not in PK unless $\rm{NP}\subseteq\rm{coNP}/\rm{poly}$.
 \begin{itemize}
  \item $e^-$-\sptr and $v^-$-\sptr
  \item $e^-$-\outtree and $v^-$-\outtree
  \item $e^-$-$d$-\cli and $v^-$-$d$-\cli
  \item $e^-$-\cli and $v^-$-\cli
 \end{itemize}
\end{corollary}

\begin{proof}
 
 We have proved already that \sptr, \outtree and $d$-\cli are NP-hard, OR-compositional and complement of monotone.
 Moreover, a tree with $k$ internal vertices is extremal for \sptr, a directed tree with $k$ leaves is extremal for \outtree and 
 the complete graph with $k$ vertices, $K_{k}$, is extremal for $d$-\cli, all of them computable in polynomial time.
 For \cli it is even simpler, as $d$-\cli is a subproblem of \cli, thus
 the nonexistence of polynomial kernels for reoptimization versions of $d$-\cli implies that such kernels also do not exist for \cli.
 \qed
\end{proof}

Observe, that we have no results relating to vertex addition. This is because, when adding vertices to a graph, there is too much freedom, on how to 
make the new vertex adjacent to a specific subset of vertices of the original graph, in this sense, it can be considered, that vertex addition is not so
much of a local modification as it is a global one. In particular, when thinking about extremal graphs, there exist problems for which
specific graphs are extremal only if vertices are added with all the other vertices as neighbors or none of them.

For AND-compositional graph problems, a similar result can be stated by constructing a reoptimization instance with
a graph that is extremal with respect to the complement problem.

\begin{theorem}\label{thm:andcompneg}
 Let \prob be a monotone (comonotone) NP-hard AND-compositional graph problem. 
 $e^-$-$\prob$ and $v^-$-$\prob$ ($e^+$-$\prob$ resp.) are not in PK unless $\rm{NP}\subseteq\rm{coNP}/\rm{poly}$.
\end{theorem}

\begin{proof}
 Let \prob be a monotone NP-hard AND-compositional graph problem and let $(G,k)$ be an instance for \prob.
 Let then $H$ be an extremal graph with respect to $\prob^c$ and $k$. This means, that
 $(H,k)\not \in \prob$, however, $(H-e,k) \in \prob$ for any edge $e$.
 
 The instance $((G\cup H,k), \bot, (G\cup (H-e),k))$ for $e^-$-$\prob$ will only be in $e^-$-$\prob$ if
 $(G,k)\in \prob$, as $(H-e,k)\in \prob$ by construction.
 
 Thus, if $e^-$-\prob would be in PK, given any instance $(G,k)$ of \prob we would be able to construct a kernel for 
 it by providing a kernel for $((G\cup H,k), \bot, (G\cup (H-e),k))$. But \prob is not in PK
 unless $\rm{NP}\subseteq\rm{coNP}/\rm{poly}$, thus proving the statement.
 
 An analogous construction proves the statement for $v^-$-$\prob$ and for problems that are comonotone.
 \qed
\end{proof}

Let us now present a problem that this theorem can be applied to, the \tw problem.
The aim of this problem is to measure how tree-like a problem is. In order to do so
we define the following structure.

\begin{definition}
 Let $G=(V,E)$ be a graph. A \emph{tree decomposition} of $G$ is a pair $D=(T,B)$, where
 $T=(V_T,E_T)$ is a tree. Let $I$ denote an arbirtary index set enumerating the vertices from $V_T$.
 Then $B$ is a labeling function $B:I\rightarrow 2^{V}$ that assigns a vertex set $X_i\subseteq V$ to
 each  index $i\in I$ (that is, to each vertex from $V_T$). These sets $X_i$ are called \emph{bags}.
 Moreover, $D$ satisfies the following properties:
 \[\bigcup_{i\in I}X_i=V,\]
 for every edge $\{u,v\}\in E$, there exists an index $i\in I$ such that $u,v\in X_i$, and for each
 $v \in V$, the bags $X_i$ containing $v$ are assigned to a subtree of $T$.
\end{definition}

The \emph{width} of $D$ is defined as $\max\{|X_i| | i\in I\}-1$, that is, the maximum size of a bag minus 1.
The \emph{treewidth} of $G$ is the minimum width over all tree decompositions of $G$, it is denoted by $\rm{tw}(G)$

Given an instance $(G,k)$ consisting of a graph $G$ and a parameter $k$, we say that $(G,k)\in \tw$ if and 
only if $\rm{tw}(G)\le k$.

This problem is NP-hard~\cite{GJ90}, AND-compositional~\cite{BDFH09}, and monotone, as we can see,
by inspection, that given a tree decomposition for a garph $G$, it is also a tree decomposition
for any $G-e$ or $G-v$ if we remove the removed vertex from the bags containing it.

In particular, as a corollary we have:

\begin{corollary}\label{cor:negex}
  $e^-$-\tw and $v^-$-\tw are not in PK unless $\rm{NP}\subseteq\rm{coNP}/\rm{poly}$.
\end{corollary}

Let us also see a concrete proof to see how one constructs the instances mentioned in the proof
of Theorem~\ref{thm:andcompneg}.

\begin{proof}
 Let $(G,k)$ be an instance for \tw. Let now $H$ be a graph with
 treewidth $k+1$ such that, for any edge $e$, $H-e$ has treewidth $k$. For instance,
 the complete graph with $k+2$ vertices $K_{k+2}$ fulfills this property.
 If $e^-$-\tw was in PK, it would be able to provide a polynomial kernel for the instance
 $((G\cup K_{k+2},k),\bot,(G\cup (K_{k+2}-e,k))$. Observe that, 
 because $K_{k+2}$ has treewidth $k+1$, the treewidth of
 $G\cup K_{k+2}>k$, thus $(G\cup K_{k+2},k)\notin \tw$  and it is valid to put $\bot$ as the second element of the instance.
 Moreover, because $K_{k+2}-e$ has treewidth $k$, $((G\cup K_{k+2},k),\bot,(G\cup (K_{k+2}-e,k))\in e^-$-$\tw$
 if and only if $(G,k)\in \tw$, thus a kernel for the reoptimization instance
 would provide us with a kernel for the initial instance.
 \qed
\end{proof}

\subsection{Other Reoptimization Compositional Problems without Polynomial Kernels}

We have seen that for every monotone, or comonotone, NP-hard compositional graph problem, two of
its reoptimization variants are not in PK. Nevertheless, finding local kernels for the other two reoptimization
variants is not trivial either.

We now provide problems where finding local kernels is just as hard as finding kernels in general, thus,
making all of their reoptimization variants as hard as the original problem.

\subsubsection{Clique}

We already know that reoptimization \cli instances parameterized by the size of the clique
do not have polynomial kernels in the case of edge and vertex deletion.

We now show that under the other two local modifications \cli is not in PK either.

\begin{theorem}\label{thm:cli}
$e^+$-$\cli$ and $v^+$-$\cli$ are not in PK unless $\rm{NP}\subseteq\rm{coNP}/\rm{poly}$.
\end{theorem}

\begin{proof}
 With a very similar technique as in the previous section, we are going to show this result by reducing \cli
 to $e^+$-$\cli$ and $v^+$-$\cli$.
 
 Let $(G,k)$ be an instance for \cli. We construct the following instance for $e^+$-$\cli$. Let $G'$ be a graph that consists
 of the graph $G$ together with $v_1$ and 
 $v_2$, two new vertices adjacent to every vertex in $G$ but not to each other. Let then $e_{1,2}$ denote
 the edge between $v_1$ and $v_2$.
 Then $((K_{k+1}\cup G',k+1),K_{k+1},(K_{k+1}\cup G'+e_{1,2}, k+2))\in e^+$-$\cli$ if and only if $(G,k)\in \cli$.
 Observe, that, because \cli is a maximization problem, it is possible to make the parameter larger.
 
 If a polynomial kernel could be computed from this constructed reoptimization instance, it would also be a polynomial kernel for the \cli instance by construction.
 So $e^+$-$\cli$ is not in PK unless \cli is, and because \cli is an NP-hard OR-compositional
 problem it does not have a kernel unless $\rm{NP}\subseteq\rm{coNP}/\rm{poly}$.
 
 As for $v^+$-$\cli$, we just need to consider, given an instance $(G,k)$ for \cli, the reoptimization instance
 $((K_k\cup G,k),K_k,(K_k\cup G + v_1,k+1))$, where $v_1$ is adjacent to every vertex in $G$, and observe again that this instance
 is in $e^+$-$\cli$ if and only if $(G,k)\in \cli$. Thus, we reach the same conclusion.
 \qed
\end{proof}

\section{Non-Compositional Problems without Polynomial Kernels. Connected Vertex Cover}
\label{sec:cvc}

One of the non-compositional problems in which reoptimization does not help us to achieve any improvement with respect to the classical parametrization techniques, is
the \covc(CVC) problem. 
A {\em connected vertex cover} of a graph $G=(V,E)$ is a subset of vertices $A\subseteq V$ that is a vertex cover of 
$G$ and such that the subgraph induced by $A$ is connected. \covc is FPT with respect to the solution size~\cite{Cy12}. Moreover, it is
conjectured that \covc does not have a polynomial Turing Kernel~\cite{HK13}.

We build a reduction from \setc that will show that even the reoptimization versions of \covc do not have a polynomial kernel
unless \setc has a polynomial compression with respect to its universe size. First we define the notion of polynomial compression.
Informally, we can think of a compression as a way to transform an instance for a problem $\prob_1$ into a kernel for a problem $\prob_2$.
This concept is a bit more general than kernelization in the sense
that it allows to show non-kernelization results for problems that
are not NP-complete:  If an NP-complete problem compresses to a
problem $X$, then the compressed instance of $X$ can be transformed
back into an instance of the original problem.  Hence, a polynomial
compression gives you automatically a polynomial kernel.

\begin{definition}\label{def:polycom}[Cygan et al.~\cite{CFK+15}]
 A \emph{polynomial compression} of a parameterized language $Q\subseteq \Sigma^*\times \N$ into a language $R\subseteq \Sigma^*$ is an algorithm
 that takes as input an instance $(x,k)\in \Sigma^*\times \N$, works in time polynomial in $|x|+k$, and returns a string $y$ such that $|y|\le p(k)$,
 for some polynomial $p$, and $y\in R$ if and only if $(x,k)\in Q$.
\end{definition}

Moreover, Dom et al. prove in \cite{DLS09} that \setc parameterized by the size of the universe does not have a polynomial compression unless $\rm{NP}\subseteq \rm{coNP}/\rm{poly}$.

We prove through a reduction that, if $e^+$-\covc had a polynomial kernel,
then \setc parameterized by the size of the universe would have a polynomial compression. Formally:

\begin{theorem}\label{thm:cvc}
$e^+$-\covc does not have a polynomial kernel unless $\rm{NP}\subseteq \rm{coNP}/\rm{poly}$.
\end{theorem}

\begin{proof}
We describe the reduction from  \setc parameterized by the size of the universe to \covc. Then we use
this reduction to prove that, if a polynomial kernel would exist for $e^+$-\covc, then we would have a polynomial compression
for \setc parameterized by the size of the universe, but such a compression is not possible, unless $\rm{NP}\subseteq \rm{coNP}/\rm{poly}$.

A \setc instance parameterized by the size of the universe is a quadruple,
$((\setU,\F, k), u)$ where $(\setU,\F, k)$ is the instance, comprised by $\setU$, the universe set,
of size $|\setU|=u$, $\F=\{F_1,\hdots F_t\}$, a family of subsets, and $k$, the solution size targeted,
and $u$ is the parameter, as defined in the introduction. We want to answer the question: Is there a subfamily of $k$ sets of \F that covers \setU?

Until now, we always considered the solution size as the parameter in all of our parameterized problems. This, however, 
is not fixed as such in the definition of kernelization, which allows us to choose the parameter with other criteria, as we do in this case.

We only consider instances for \setc where $k\le u$, because otherwise the solution is trivial. This is because any
subset $F_i$ in the optimal solution should cover at least one element in \setU that is not covered by any other selected
subset. Otherwise, the subset could be trivially removed and the solution would be smaller.

Let us first show the following reduction from \setc to \covc. Given a \setc instance $((\setU,\F, k), u)$,
we construct an instance for \covc as shown in Fig.~\ref{fig:cscsol1}.

We create a grid of vertices $u_{i,j}$ where $i=1,\hdots ,k+2$ and $j=0,\hdots, u$. 
Each of these vertices has an attached leaf $u'_{i,j}$.
Each one of the columns 1 to $u$ of the original grid represents one of the elements of the universe set in the \setc instance. Column 0 is an additional column which
can be viewed as an extra element added to the set.

We add a row of vertices $f_1,\hdots, f_t$ such that each vertex $u_{i,j}$ of the column $j$, 
will be connected to $f_\ell$ if and only if $j\in F_\ell$.
We also add a vertex $x$ which is connected to all the first column $u_{i,0}$ for all $i=1,\hdots,k+1$,
except for $u_{k+2,0}$. We can think of $x$ as an extra set in the family of subsets $\F$, containing only
the new element of the set. The edge $(x,u_{k+2,0})$ will be added in the reoptimization step (dashed edge in Fig.~\ref{fig:cscsol1}).

We add a column $v_1,\hdots, v_{k+2}$ such that each vertex $u_{i,j}$ of row $i$, will be connected to
$v_i$, as represented in Fig.~\ref{fig:cscsol1}.
Finally, we add two vertices $f$ and $y$, $f$ neighboring $f_\ell$, for all $\ell=1,\hdots,t$, and also $x$ and $y$, and $y$ additionally neighboring $v_i$, for 
all $i=1,\hdots,k+2$.

Now we will use this reduction to prove Theorem~\ref{thm:cvc}.

 Given a non-trivial \setc instance $((\setU,\F, k), u)$, we construct an $e^+$-\covc instance as follows.
 Take first the graphs $G$, $G+e$ constructed by the 
 reduction described above. 
 Now, this instance is not complete unless we provide the appropriate parameters for $G$ and $G+e$ and a solution for $G$.
 
 We will now construct two optimal solutions for $G$.
 
 To select a connected vertex cover in the graph, first observe that, if the grid vertex $u_{i,j}$ is not part of the connected vertex cover, then, even if all the other vertices in the graph were in the cover,
the cover would not be connected.

Thus, we select all the vertices $u_{i,j}$ of the grid (\ie, $(k+2)(u+1)$ vertices in total). This covers all leaf edges $(u_{i,j},u'_{i,j})$, all edges $(u_{i,j},f_\ell)$ and
all edges $(u_{i,j},v_{i})$. It does not cover, however, the edges $(f_\ell,f)$, $(v_{i}, y)$ and $(y,f)$, and it is not connected.

The vertex $f$  needs to be taken to cover the edges $(f_\ell,f)$ because, again, taking every $f_\ell$ would make the vertex cover too large with respect to $u$.
This is because $t$ has to be at least $k+1$ for the instance to be nontrivial.
Moreover, to connect the $u_{i,j}$, we have two options:
\begin{enumerate}
 \item Take all $v_{i}$ and $y$: With these vertices we cover the remaining 
edges and we obtain a CVC of size $(k+2)(u+1)+k+2+2=(k+2)(u+2)+2$. We will name this solution $S_1$ and we will
also name $c=(k+2)(u+2)$ making the size of this solution $c+2$. (See Fig.~\ref{fig:cscsol1})

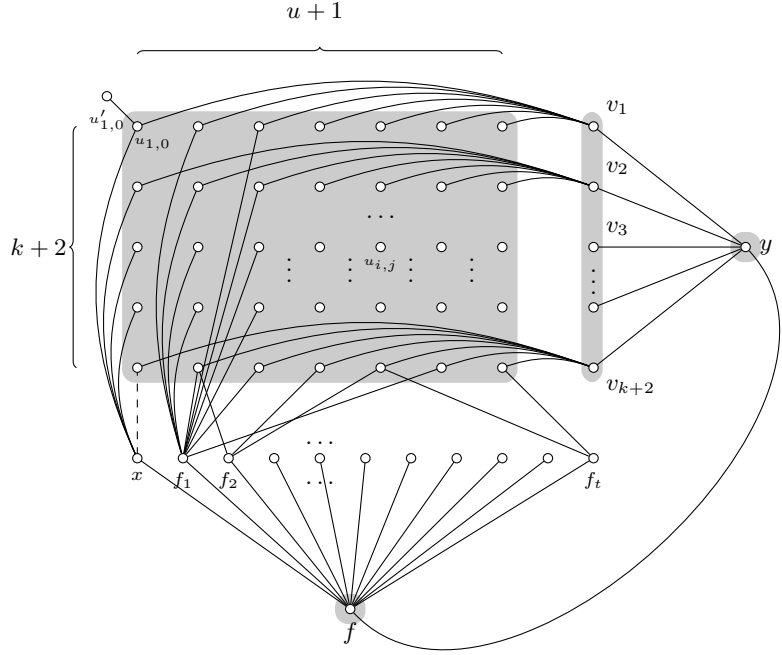
\begin{figure}[t]
\begin{center}

\begin{tikzpicture}[scale=0.8,vertex/.style={draw,circle,minimum size=3.5pt,inner sep=1pt,fill=white}]
\tikzfading[name=fade out, inner color=transparent!0,
         outer color=transparent!100]
         
\fill [black!20,rounded corners=5pt] (-0.25,-0.25) rectangle (6.25,4.25);
\fill [black!20,rounded corners=5pt] (7.3,-0.25) rectangle (7.65,4.25);
\fill [black!20,rounded corners=5pt] (3.25,-4.25) rectangle (3.75,-3.75);
\fill [black!20,rounded corners=5pt] (9.75,1.75) rectangle (10.25,2.25);

\foreach \y in {1,...,4}{
  \draw (0,\y) edge [bend right=25] (0,-1.5);
}
\draw (0,0) edge [dashed] (0,-1.5);

\draw (-0.5,4.5) edge (0,4);
\node [vertex, label=below:\tiny{$u'_{1,0}$}] at (-0.5,4.5) {};
\node at (0.25,3.7) {\tiny{$u_{1,0}$}};

\node [vertex,label=below:$f$] at (3.5,-4)(F){};
\foreach \x in {0,...,10}{
\draw (F) edge (\x * 0.75,-1.5);
}
\node[vertex, label=below:\scriptsize{$x$}] at (0,-1.5)(f_0){};
\foreach \x in {3,...,8}{
\node[vertex] at (\x * 0.75,-1.5)(f_\x){};
}
\node[vertex, label=below:\scriptsize{$f_1$}] at ( 0.75,-1.5)(f_1){};
\node[vertex, label=below:\scriptsize{$f_2$}] at (1.5,-1.5)(f_2){};
\node[vertex] at (6.75,-1.5){};
\node[vertex, label=below:\scriptsize{$f_{t}$}] at (7.5,-1.5)(f_t){};
\node at (3,-1.9) {$\hdots$};

\draw (f_1) edge (1,0);
\draw (f_1) edge [bend left=20](1,1);
\draw (f_1) edge [bend left=20](1,2);
\draw (f_1) edge [bend left=20](1,3);
\draw (f_1) edge [bend left=20](1,4);
\draw (f_1) edge (2,0);
\draw (f_1) edge (2,1);
\draw (f_1) edge (2,2);
\draw (f_1) edge (2,3);
\draw (f_1) edge (2,4);
\node at (2.5,1.75) {$\vdots$};
\node at (3.5,1.75) {$\vdots$};
\node at (4.5,1.75) {$\vdots$};
\node at (5.5,1.75) {$\vdots$};
\draw (f_1) edge (5,0);
\draw (f_2) edge (1,0);
\draw (f_2) edge (3,0);
\draw (f_2) edge (4,0);
\node at (3,-1.25) {$\hdots$};
\draw (f_t) edge (6,0);
\draw (f_t) edge (4,0);

\node [vertex, label=right:$y$] at (10,2)(H){};
\foreach \x in {0,...,4}{
\draw (H) edge (7.5,\x);
}

\foreach \x in {0,...,6}{
\draw (7.5,4) edge [bend right=20] (\x,4);
\draw (7.5,3) edge [bend right=20] (\x,3);
\draw (7.5,0) edge [bend right=20] (\x,0);
}
\node at (4,2.5) {$\hdots$};

\foreach \y in {1,...,3}{
\node[vertex, label=above right:$v_{\y}$] at (7.5, 5-\y){};
}
\node [vertex, label=above:$\vdots$] at (7.5,1){};
\node [vertex, label=below right:$v_{k+2}$] at (7.5,0){};

\draw (F) edge [bend right=90] (H);

\foreach \x in {0,...,6}{
  \foreach \y in {0,...,4}{
    \node[vertex] at (\x,\y){};
  }
}
\node at (4,1.7) {\tiny{$u_{i,j}$}};

\draw [decoration={brace},decorate] (-1,0) -- (-1,4) 
node [pos=0.57,anchor=north,xshift=-0.50cm] {$k+2$}; 
\draw [decoration={brace},decorate] (0,5.2) -- (6,5.2) 
node [pos=0.57,anchor=north,xshift=-0.40cm, yshift=0.8cm] {$u+1$};
\end{tikzpicture}

\end{center}

\caption{Drawing of the reduction graph. Here the leaves $u'_{i,j}$ are left out of the drawing for clarity.
 The connected vertex cover for this graph using the side vertices $v_m$ is shaded in grey.}
\label{fig:cscsol1}
\end{figure}
 
 \item Take a selection of $F_\ell$ that covers all columns and also take $v_{k+2}$ in order to connect the vertex $u_{k+2,0}$. Finally, we also take $y$ to
 cover the edges $(v_i,y)$.
 This makes a total of $(k+2)(u+1)+SCsol+1+3$ where $SCsol+1$ stands for the size of a Set Cover solution for $(\setU,\F)$ together with the vertex $x$ and $+~3$ stands for the vertices
 $f$, $v_k$ and $y$. We will name this solution $S_2$. (See Fig.~\ref{fig:cscsol2})
 
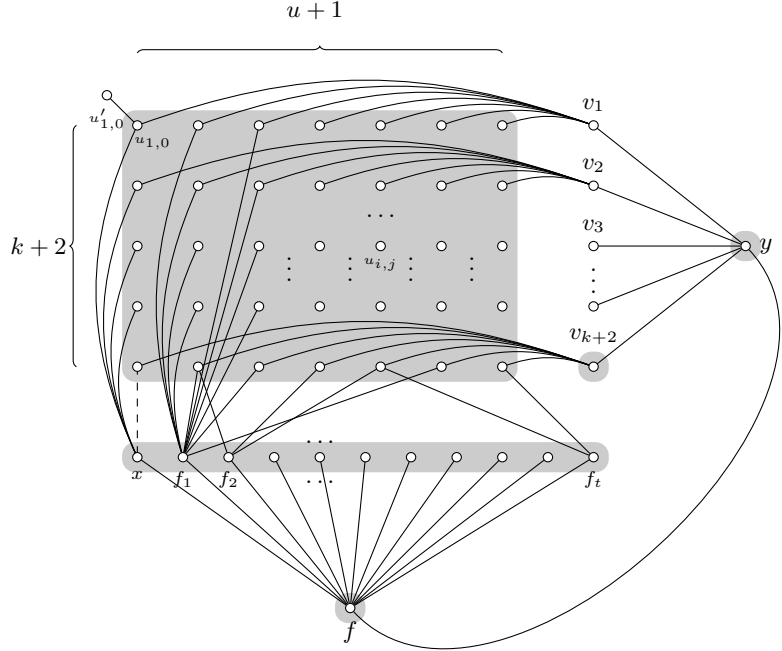
\begin{figure}[t]
\begin{center}

 \begin{tikzpicture}[scale=0.8,vertex/.style={draw,circle,minimum size=3.5pt,inner sep=1pt,fill=white}]
 \tikzfading[name=fade out, inner color=transparent!0,
         outer color=transparent!100]
         
 \fill [black!20,rounded corners=5pt] (-0.25,-0.25) rectangle (6.25,4.25);
\fill [black!20,rounded corners=5pt] (-0.25,-1.25) rectangle (7.75,-1.75);
\fill [black!20,rounded corners=5pt] (3.25,-4.25) rectangle (3.75,-3.75);
\fill [black!20,rounded corners=5pt] (9.75,1.75) rectangle (10.25,2.25);
\fill [black!20,rounded corners=5pt] (7.25,-0.25) rectangle (7.75,0.25);

\foreach \y in {1,...,4}{
  \draw (0,\y) edge [bend right=25] (0,-1.5);
}
\draw (0,0) edge [dashed] (0,-1.5);

\draw (-0.5,4.5) edge (0,4);
\node [vertex, label=below:\tiny{$u'_{1,0}$}] at (-0.5,4.5) {};
\node at (0.25,3.7) {\tiny{$u_{1,0}$}};

\node [vertex,label=below:$f$] at (3.5,-4)(F){};
\foreach \x in {0,...,10}{
\draw (F) edge (\x * 0.75,-1.5);
}
\node[vertex, label=below:\scriptsize{$x$}] at (0,-1.5)(f_0){};
\foreach \x in {3,...,8}{
\node[vertex] at (\x * 0.75,-1.5)(f_\x){};
}
\node[vertex, label=below:\scriptsize{$f_1$}] at ( 0.75,-1.5)(f_1){};
\node[vertex, label=below:\scriptsize{$f_2$}] at (1.5,-1.5)(f_2){};
\node[vertex] at (6.75,-1.5){};
\node[vertex, label=below:\scriptsize{$f_{t}$}] at (7.5,-1.5)(f_t){};
\node at (3,-1.9) {$\hdots$};

\draw (f_1) edge (1,0);
\draw (f_1) edge [bend left=20](1,1);
\draw (f_1) edge [bend left=20](1,2);
\draw (f_1) edge [bend left=20](1,3);
\draw (f_1) edge [bend left=20](1,4);
\draw (f_1) edge (2,0);
\draw (f_1) edge (2,1);
\draw (f_1) edge (2,2);
\draw (f_1) edge (2,3);
\draw (f_1) edge (2,4);
\node at (2.5,1.75) {$\vdots$};
\node at (3.5,1.75) {$\vdots$};
\node at (4.5,1.75) {$\vdots$};
\node at (5.5,1.75) {$\vdots$};
\draw (f_1) edge (5,0);
\draw (f_2) edge (1,0);
\draw (f_2) edge (3,0);
\draw (f_2) edge (4,0);
\node at (3,-1.25) {$\hdots$};
\draw (f_t) edge (6,0);
\draw (f_t) edge (4,0);

\node [vertex, label=right:$y$] at (10,2)(H){};
\foreach \x in {0,...,4}{
\draw (H) edge (7.5,\x);
}

\foreach \x in {0,...,6}{
\draw (7.5,4) edge [bend right=20] (\x,4);
\draw (7.5,3) edge [bend right=20] (\x,3);
\draw (7.5,0) edge [bend right=20] (\x,0);
}
\node at (4,2.5) {$\hdots$};

\foreach \y in {1,...,3}{
\node[vertex, label=above:$v_{\y}$] at (7.5, 5-\y){};
}
\node [vertex, label=above:$\vdots$] at (7.5,1){};
\node [vertex] at (7.5,0){};
\node at (7.5,0.5) {$v_{k+2}$};
\draw (F) edge [bend right=90] (H);

\foreach \x in {0,...,6}{
  \foreach \y in {0,...,4}{
    \node[vertex] at (\x,\y){};
  }
}
\node at (4,1.7) {\tiny{$u_{i,j}$}};

\draw [decoration={brace},decorate] (-1,0) -- (-1,4) 
node [pos=0.57,anchor=north,xshift=-0.50cm] {$k+2$}; 
\draw [decoration={brace},decorate] (0,5.2) -- (6,5.2) 
node [pos=0.57,anchor=north,xshift=-0.40cm, yshift=0.8cm] {$u+1$};
\end{tikzpicture}

\end{center}

\caption{Second CVC option using a SC for $\F$. Observe, that not all of the vertices $f_i$, are part of the cover, only a selection of them.}
\label{fig:cscsol2}
\end{figure}
 
\end{enumerate}
Mixing these two strategies is not a good option because, for any column $j$ missed by the set cover, in order to connect every vertex $u_{i,j}$ to 
the rest of the vertex cover, we would need to take all vertices $v_i$, $1\le i\le k+2$, into the cover, rendering the selection of any $f_{\ell}$ pointless.

Both solutions are of the same size if and only if the size of the optimal set cover for $\F$ is $k$.

Let us assume that, if the size of the optimal set cover for \F is $k$, then the \covc instance consists of $((G,c+2),S_1,(G+e,c+1))$ 
(both solutions $S_1$ and $S_2$ are optimal if this is the case, so $S_1$ is a valid choice).

In order to solve the instance for $G+e$, we have to consider that,
once the dashed edge is added to the graph, $S_2$ is still a solution and furthermore it can be reduced by one vertex by not taking $v_{k+2}$, i.e., $S_2-v_{k+2}$ is a solution for $G+e$. This vertex
is not needed anymore because now the edge $e$ already connects $u_{k+2,0}$, via the vertex $x$, to the rest of the vertex cover. 
However, it is easy to see that one cannot remove any vertex of $S_1$. This means that only $S_2-v_{k+2}$ is optimal, once the new edge is added.

We can conclude that, if we could get a polynomial kernel for $(G+e,c+1)$ given the instance $((G,c+2),S_1,(G+e,c+1))$
then we would be able to find an instance $(G',k')$ polynomial in $k\cdot u$ (and thus, because $k\le u$, polynomial in $u$) such that $(G',k')\in \rm{CVC}$
if and only if $((\setU,\F, k), u)\in \rm{SC}$ (\setc). This by definition means that we would have a polynomial compression for \setc parameterized by the size
of the universe. But \setc parameterized by the size of the universe does not admit a polynomial compression unless $\rm{NP}\subseteq \rm{coNP}/\rm{poly}$.
\qed
\end{proof}

\section{Reoptimization and Vertex Cover}
\label{sec:vc}

Another case where we observe the power of reoptimization in parameterized problems is in the \verc(VC) problem.
\verc is a problem in PK, whose best known polynomial kernel is of size $2k$ using linear programming \cite{NT75}. 
However, using crown decomposition only
allows us to achieve a kernel of size $3k$ in the classical setting~\cite{AFLS07} (see also~\cite{FG06}).
Very recently and independent of our work, there have been attempts to reduce the size of the crown decomposition kernel by refining the method, like in~\cite{LZ18}.
We present here a way to achieve a kernel of size $2k$ using reoptimization and crown decomposition.

First we define the problem. A vertex cover of a graph $G=(V,E)$ is a subset $A\subseteq V$ such that every edge is covered, \ie, every edge $e\in E$ is
incident to a vertex $v\in A$. As a parameterized problem, we say $(G,k)\in \rm{VC}$ if there exists a vertex cover of $G$ of size $k$ 
or smaller.

The crown decomposition is a structure in a graph that can be defined as follows. (It is shown schematically in
Fig.~\ref{fig:crowndecomp}.) 

\begin{definition}\label{def:crown-decomp}
  Let $G=(V,E)$ be a graph. A \emph{crown decomposition}
  \index{crown decomposition} of $G$ is a partition of $V$ into three sets
  $C$, $H$, and $R$ satisfying the following properties.
  \begin{enumerate} 
  \item $C$ is a non-empty independent set in $G$,
  \item There are no edges between $C$ and $R$,
  \item The set of edges between $C$ and $H$ contains a matching $M$ of
    size $|H|$, we also say that $M$
    \emph{saturates}\index{matching!saturating} $H$.
  \end{enumerate}

  We call $C$ the \emph{crown}, $H$ the
  \emph{head}, and $R$ the
  \emph{rest} of the crown decomposition. 
\end{definition}

\begin{figure}[t]
\begin{center}
\begin{tikzpicture}[scale=0.8, point/.style={draw,circle,minimum size=3.5pt,inner sep=1pt,fill=white},path/.style={very thick},arrow/.style={-stealth, densely dashed},font=\scriptsize]
  \draw[white, fill=black!20] (-1.5,-3.2) ellipse (2cm and 0.75cm);
  \node[point] (c1) at (0,0) {};
  \node[point] (c2) at (-1,0) {};
  \node[point] (c3) at (-2,0) {};
  \node[point] (c4) at (-3,0) {};
  \node[point] (h1) at (-0.5,-1.5) {};
  \node[point] (h2) at (-1.5,-1.5) {};
  \node[point] (h3) at (-2.5,-1.5) {};
  \node[point] (b1) at (0,-3) {};
  \node[point] (b2) at (-0.75,-3.5) {};
  \node[point] (b3) at (-1.5,-3) {};
  \node[point] (b4) at (-2.25,-3.5) {};
  \node[point] (b5) at (-3,-3) {};
  \node (B) at (1,-3.2){$R$};
  \node (H) at (1,-1.5) {$H$};
  \node (C) at (1,0) {$C$};
  \draw
    (h1) edge[path] (c1)
    (h1) edge (h2)
    (h1) edge (c2)
    (h1) edge (c3)
    (h2) edge (c2)
    (h2) edge[path] (c3)
    (h2) edge (c4)
    (h3) edge[path] (c4)
    (h3) edge (c3)
    (h1) edge (b1)
    (h1) edge (b2)
    (h2) edge (b2)
    (h2) edge (b3)
    (h2) edge (b4)
    (h3) edge (b5)
    (b1) edge (b2)
    (b1) edge (b4)
    (b2) edge (b3)
    (b3) edge (b4)
    (b3) edge (b5)
    (b4) edge (b5);
\end{tikzpicture}
\end{center}
\caption{Example of a crown decomposition of a graph.}
\label{fig:crowndecomp}
\end{figure}
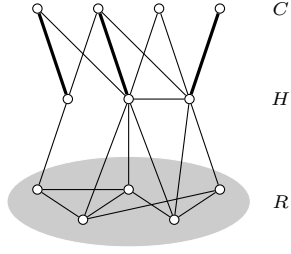

The crown lemma tells us under which conditions crown decompositions exist, and is the basis for kernelization using crown decomposition.
\begin{lemma}[\cite{AFLS07}]\label{lm:crown}
 Let $G$ be a graph without isolated vertices and with at least $3k+1$ vertices. There is a polynomial-time algorithm that either finds
 a matching of size $k+1$ in $G$ or finds a crown decomposition of $G$.
\end{lemma}

This lemma allows us to reduce any \verc instance to size at most $3k$.
This is because, given a graph of size larger than $3k$, we either find a matching of size $k+1$ or a crown decomposition of $G$.
Given a crown decomposition of $G$ into $H$, $C$ and $R$, take the maximum matching between $H$ and $C$. This matching provides
proof that any vertex cover for $G$ will need at least $|H|$ vertices to cover the vertices in the matching. Thus, we may reduce 
an instance $(G,k)$ to an instance $(G-(H\cup C),k-|H|)$. 

Let us consider the \verc problem under edge addition, \ie, $e^+$-VC. Given an (optimal) $k$-vertex cover for a graph $G$, we will give a kernel of size $2k$
for $(G+e,k)$ using crown decomposition.
First of all, let $A\subseteq V$ be the optimal vertex cover of $G$,  and $B\subseteq V$ the rest of the vertices in $G$.
Then, $|A|=k$ and $|B|=n-k$.
Observe, that there might be edges between vertices of $A$ but not between vertices of $B$, otherwise $A$ would not be a vertex cover.

\begin{figure}[t]
 \begin{center}
  \subfloat[First crown decomposition\label{fig:cd1}]{
    \begin{tikzpicture}[point/.style={draw, ellipse, minimum height=8mm, align=center, fill=white},path/.style={line width=0.7mm},font=\scriptsize]
      \draw[white] (-2,0)--(4,0)--(4,-4.5)--(-2,-4.5);
      \fill[black!20, rounded corners=5pt] (1.7, 0.5)--(2.8, 0.5)--(2.8,-5.5)--(1.7,-5.5)--(1.7, -2.5)-- (2.5,-2.5)--(2.5,-1)--(1.7,-1)--cycle;
      \node (c) at (2.5,-4){\normalsize{$\bf{C_1}$}};
      
      \fill[black!20, rounded corners=5pt] (0.3, 0.5)--(-0.8, 0.5)--(-0.8,-4)--(0.3,-4)--(0.3, -2.5)-- (-0.5,-2.5)--(-0.5,-1)--(0.3,-1)--cycle;
      \node (h) at (-0.5,-3){\normalsize{$\bf{H_1}$}};

      \node[point,label=below:$A^3_M$] (h1) at (0,0) {};
      \node[point,label=below:$A^2_M$] (h2) at (0,-1.5) {};
      \node[point,label=below:$B^3_M$] (c1) at (2,0) {};
      \node[point,label=below:$B^2_M$] (c2) at (2,-1.5) {};
      \node[point,label=below:$A^1_M$] (v1) at (0,-3) {};
      \node[point,label=below:$B^1_M$] (v2) at (2,-3) {};
      \node[point,label=below:$A_{\overline{M}}$] (v3) at (0,-4.5) {};
      \node[point,label=below:$B_{\overline{M}}$] (v4) at (2,-4.5) {};
      \node (vm) at (0,-6){\normalsize{$\bf{A}$}};
      \node (i) at (2,-6) {\normalsize{$\bf{B}$}};
      \draw
	(h1) edge[path] (c1)
	(h1) edge[dashed] (c2)
	(v1) edge[dashed] (c1)
	(v1) edge[dashed] (c2)
	(h2) edge[path] (c2)
	(v1) edge (v4)
	(v3) edge (c2)
	(v1) edge[path] (v2);
    \end{tikzpicture}
  }
  \subfloat[Second crown decomposition\label{fig:cd2}]{
    \begin{tikzpicture}[point/.style={draw, ellipse, minimum height=8mm, align=center, fill=white},path/.style={line width=0.7mm},font=\scriptsize]
      \draw[white] (-2,0)--(4,0)--(4,-5.5)--(-2,-5.5);
       \fill[black!20, rounded corners=5pt] (2.8, -2.5)--(2.8,-5.5)--(1.7,-5.5)--(1.7,-2.5)--cycle;
      \node (c) at (2.5,-4){\normalsize{$\bf{C_2}$}};
      
      \fill[black!20, rounded corners=5pt] (-0.8, -2.5)--(-0.8,-4)--(0.3,-4)-- (0.3, -2.5) --cycle;
      \node (h) at (-0.5,-3){\normalsize{$\bf{H_2}$}};

      \draw[white] (-2,0)--(4,0)--(4,-4.5)--(-2,-4.5);
      \node[point,label=below:$A^3_M$] (h1) at (0,0) {};
      \node[point,label=below:$A^2_M$] (h2) at (0,-1.5) {};
      \node[point,label=below:$B^3_M$] (c1) at (2,0) {};
      \node[point,label=below:$B^2_M$] (c2) at (2,-1.5) {};
      \node[point,label=below:$A^1_M$] (v1) at (0,-3) {};
      \node[point,label=below:$B^1_M$] (v2) at (2,-3) {};
      \node[point,label=below:$A_{\overline{M}}$] (v3) at (0,-4.5) {};
      \node[point,label=below:$B_{\overline{M}}$] (v4) at (2,-4.5) {};
      \node (vm) at (0,-6){\normalsize{$\bf{A}$}};
      \node (i) at (2,-6) {\normalsize{$\bf{B}$}};
      \draw
	(h1) edge[path] (c1)
	(h1) edge[dashed] (c2)
	(v1) edge[dashed] (c1)
	(v1) edge[dashed] (c2)
	(h2) edge[path] (c2)
	(v1) edge (v4)
	(v3) edge (c2)
	(v1) edge[path] (v2);
      
    \end{tikzpicture}
  }
\end{center}
\caption{Partition of $G$. Bold edges belong to the matching, solid edges indicate that there exist edges between the two subsets and dashed edges indicate that edges might exist
between the two groups. Any non-edge between $A$ and $B$ indicates that edges do not exist between the two subsets by construction. The shaded subsets indicate the head and crown of 
two possible crown decompositions parting from this partition of $G$, the non-shaded subsets are the rest sets $R_1$ and $R_2$, respectively.}
\label{fig:vercov}
\end{figure}
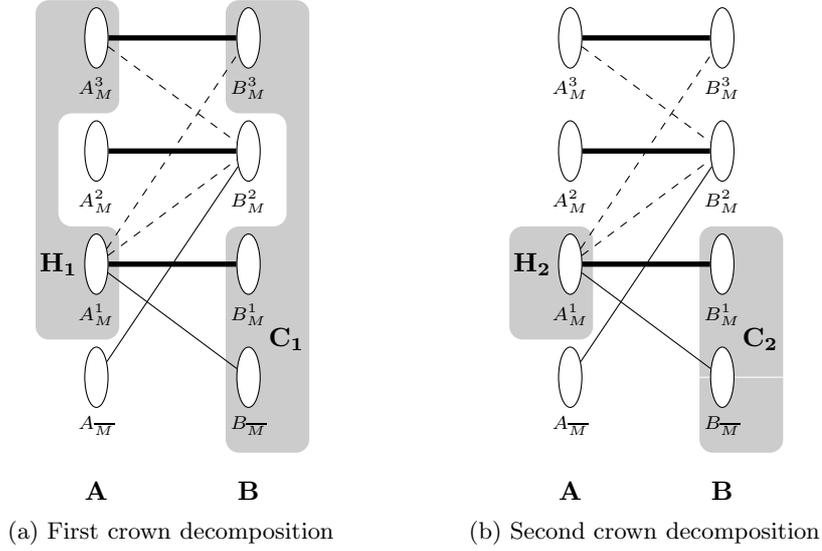
Let us pick $M$ to be a maximal matching between $A$ and $B$.
We will now partition $A$ and $B$ further into subsets according to their adjacencies (see Fig.~\ref{fig:vercov}). First, we consider the vertices of $A$ and $B$ that are 
not part of the matching, let us call these vertices $A_{\overline{M}}$ and $B_{\overline{M}}$ respectively. The vertices of these two subsets do not share edges because otherwise $M$ would not 
be maximal. The vertices of $A$ and $B$ that are part of the matching will be $A_{M}$ and $B_{M}$, respectively. 
Now let $A_{M}^1$ be the matched vertices in $A$ that have an alternating path to at least one vertex in $B_{\overline{M}}$ and
let $B_{M}^1$ be the vertices matched to those from $A_{M}^1$. Let then $B_{M}^2$ be the vertices in $B$ that have an alternating path to at least one vertex in $A_{\overline{M}}$ and 
let $A_{M}^2$ be the vertices matched to those from $B_{M}^2$. These four subsets have no intersection because otherwise there would be an augmenting path
starting in $B_{\overline{M}}$ through an alternating path to $v\in A_{M}^1\cap A_{M}^2$ and  through its matched vertex in $B_{M}^2$ and another alternating path to $A_{\overline{M}}$, contradicting
the maximality of $M$. Let then $A_{M}^3$ and $B_{M}^3$ be the rest of the matched vertices in $A$ and $B$, respectively.

Observe, through Fig.~\ref{fig:cd1}, that the following is a valid crown decomposition for $G$: 
\[B_{\overline{M}}\cup B_{M}^1 \cup B_{M}^3=C_1\text{, }A_M^1\cup A_M^3=H_1\text{ and }A_{\overline{M}}\cup A_M^2\cup B_M^2=R_1.\]
This is true because $B$ is an independent set by construction, and there are no edges between $C_1$ and $R_1$ because $B_{\overline{M}}$ and $B_{M}^1$
have edges 
neither to $A_{\overline{M}}$ nor to $A_M^2$ and neither does $B_{M}^3$.
Moreover, as depicted in Fig.~\ref{fig:cd2},
\[B_{\overline{M}}\cup B_{M}^1=C_2\text{, }A_M^1=H_2\text{ and }B_{M}^3\cup A_{M}^3\cup R_1=R_2,\] 
is also a valid crown decomposition, as $B_{\overline{M}}$ and  $B_{M}^1$ also have no edges to $A_{M}^3$.

In the second crown decomposition, $|R_2|< 2k$ because $|A|\le k$ and thus $|B_M|\le k$, meaning that $|B_{\overline{M}}|\ge n-2k$ and $B_{\overline{M}}$ is always part of $C$. 
If $H$ is empty, it means that there is a set of isolated vertices in $G$ of size  at least $n-2k$ and we can erase them.

Now we prove using these two crown decompositions that we can construct a crown decomposition of size $2k$ for any $G+e$ under edge addition.
If the new edge is incident to an isolated vertex, we can use 
the following reduction: if $G+e-\{\text{isolated vertices}\}$ contains a leaf, add the vertex adjacent to the leaf to the cover and use the rest of the graph as a ($2k-1$)-sized kernel.
Thus, except in this special case, $|C_2|\ge n-2k+1$, $|H_2|\ge 1$, and $|R_2|< 2k-1$.

If the new edge $e$ is adjacent to any vertex in $A$, the optimal vertex cover for $G$ is
also an optimal vertex cover for $G+e$ and thus the problem is solved. If $e$ is adjacent to two vertices $u$ and $v$ in $B$, we make the following case distinction.
\begin{description}
 \item[\bf{Case 1}\namedlabel{case1}{Case 1}] $u,v\in B^2_{M}\cup B^3_{M}$: $C_2,\ H_2$ and $R_2$ are also a crown decomposition for $G+e$.
 \item[\bf{Case 2}\namedlabel{case2}{Case 2}] $u,v\in B_{\overline{M}}$: Set $H=H_1\cup u$ and $C=C_1-u$ and $R=R_1$. The new edge $e$ provides the matching between $u$ and $v$, so $H$, $C$ and $R$ as defined are a crown decomposition.
 \item[\bf{Case 3}\namedlabel{case3}{Case 3}] $u\in  B_{\overline{M}}\text{, }v\in B^2_{M}\cup B^3_{M}$: Set $R=R_2\cup u$, $C=C_2-u$ and $H=H_2$. This provides a valid crown decomposition, as $e$ will be left inside $R$.
 \item[\bf{Case 4}\namedlabel{case4}{Case 4}] $u\in  B^1_{M}\text{, }v\in B^2_{M}\cup B^3_{M}$: There is always an alternating path between the vertex matched to $u$ and $B_{\overline{M}}$. This path provides
 an alternative maximum matching $M'$ that does not use $u$. Thus, we are in \ref{case3}.
 \item[\bf{Case 5}\namedlabel{case5}{Case 5}] $u\in B^1_{M}\text{, }v\in B^1_{M}\cup B_{\overline{M}}$: Using the same technique as in \ref{case4}, we can assume $v\in B_{\overline{M}}$. If there is an alternating
 path from $u$ to $B_{\overline{M}}-v$, we can, again, use the same technique as in \ref{case4} and we are in \ref{case2}. Otherwise, every alternating path from $u$ to $B_{\overline{M}}$
 leads exclusively to $v$. Meaning that there is a set of vertices $B_v\subseteq B^1_M$ and $A_v\subseteq A^1_M$ that do not contain edges to  any other vertex in $B^1_M$ or $B_{\overline{M}}$.
 Redefining $R= v\cup B_v\cup A_v\cup R_2$, $H=H_2-A_v$ and $C=C_2-(B_v\cup v)$, we have a valid crown decomposition.
\end{description}

For every crown decomposition we defined, the set $R$ contains not more than $2k$ vertices, thus these decompositions provide a kernel of size $2k$ for $e^+$-\verc.

If we consider \verc with other local modifications, we observe it is easy to use the same technique. Adding vertices and deleting vertices or edges allows us to use exactly the same 
crown decomposition and similar techniques to find kernels of size at most $2k$.

\section{Conclusions and Further Research}

We presented examples of problems that do not have polynomial kernels under standard complexity-theoretic assumptions, but whose reoptimization versions have polynomial kernels.
We also presented an example, where a kernel using the same technique is smaller in the reoptimization version of the problem. We finally presented a reduction proving that
there are problems and local modifications, for which the complexity does not decrease when considering reoptimization.

In conclusion, there are problems that are easier under reoptimization conditions and problems that are not. We hope that further research will help us to better understand how 
much information neighboring solutions are providing and when this
information is helpful.

\section*{Acknowledgements}

We thank Fabian Frei, Janosch Fuchs, Juraj Hromkovi\v{c}, Dennis Komm, Tobias M\"{o}mke and Walter Unger
for helpful comments and discussions.

\bibliographystyle{plain}

\begin{thebibliography}{10}
  \bibitem{AEF+15}
    F.~N.~Abu-Khzam, J.~Egan, M.~R.~Fellows, F.~A.~Rosamond,
    and P.~Shaw.
    \newblock On the parameterized complexity of dynamic problems.
    \newblock \emph{Theoretical Computer Science} 607, pp.~426--434, 2015.
  \bibitem{AMW17}
    J.~Alman, M.~Mnich, and V.~Vassilevska Williams.
    \newblock Dynamic parameterized problems and algorithms.
    \newblock In \emph{Proceedings of the 44th Intermational Colloquium on
    Automata, Languages and Programming (ICALP~2017)}, LIPiCS, Dagstuhl
    Publishing, pp.~41:1--41:16, 2017.
    
  \bibitem{AFLS07}
    F.~N.~Abu-Khzam, M.~R.~Fellows, M.~A.~Langston, and H.~W.~Suters.
    \newblock Crown Structures for Vertex Cover Kernelization.
    \newblock {\em Theory of Computing Systems}, 41(3):411--430, 2007.

  \bibitem{ABS03}
    C.~Archetti, L.~Bertazzi, and M.~G.~Speranza.
    \newblock Reoptimizing the traveling salesman problem.
    \newblock {\em Networks}, 42(3):154--159, 2003.

  \bibitem{AEM+06}
    G.~Ausiello, B.~Escoffier, J.~Monnot, and V.~Paschos.
    \newblock Reoptimization of minimum and maximum traveling salesman’s tours.
    \newblock In {\em Proc.\ of the 10th Scandinavian Workshop on Algorithm Theory (SWAT~2006)}, LNCS 4059, pp.~196--207, Springer-Verlag, 2006.


  \bibitem{BFH+06}
    H-J.~B{\"{o}}ckenhauer, L.~Forlizzi, J.~Hromkovi{\v{c}}, J.~Kneis, J.~Kupke, G.~Proietti, and P.~Widmayer.
    \newblock Reusing optimal {TSP} solutions for locally modified input instances.
    \newblock In {\em Proc.\ of the 4th {IFIP} International Conference on Theoretical Computer Science (IFIP TCS~2006)}, pp.~251--270, 2006.

  \bibitem{BHM+08}
    H-J.~B{\"{o}}ckenhauer, J.~Hromkovi{\v{c}}, T.~M{\"o}mke, and P.~Widmayer.
    \newblock On the hardness of reoptimization.
    \newblock In {\em Proc.\ of the 34th Conference on Current Trends in Theory and Practice of Computer Science (SOFSEM~2008)}, LNCS 4910, pp.~50--65, Springer-Verlag, 2008.

  \bibitem{BDFH09}
    H.~L.~Bodlaender, R.~G.~Downey, M.~R.~Fellows, and D.~Hermelin.
    \newblock On problems without polynomial kernels.
    \newblock {\em Journal of Computer and System Sciences}, 75(8):423 - 434, 2009. 
  
  \bibitem{CLRS09}
    T.~H.~Cormen, C.~E.~Leiserson, R.~L.~Rivest, and C.~Stein.
    \newblock {\em Introduction to Algorithms, Third Edition}.
    \newblock The MIT Press, 2009.
    
  \bibitem{Cy12}
    M.~Cygan.
    \newblock Deterministic Parameterized Connected Vertex Cover.
    \newblock In {\em Proc.\  of Algorithm Theory  (SWAT~2012)}, pp.~95--106, Springer-Verlag, 2012.
    
  \bibitem{CFK+15}
    M.~Cygan, F.~V.~Fomin, L.~Kowalik, D.~Lokshtanov, D.~Marx, M.~Pilipczuk, M.~Pilipczuk, and S.~Saurabh. 
    \newblock \emph{Parameterized Algorithms}.
    \newblock Springer, 2015.
    
  \bibitem{DT09}
    J.~Daligault, S.~Thomass{\'e}.
    \newblock On Finding Directed Trees with Many Leaves.
    \newblock In {\em Parameterized and Exact Computation}, pp.~86--97, Springer-Verlag, 2009.
    
  \bibitem{DLS09}
    M.~Dom, D.~Lokshtanov, s.~Saurabh.
    \newblock Incompressibility Through Colors and IDs.
    \newblock In {\em Proc.\ of Automata, Languages and Programming (ICALP 2009)}, LNCS 5555, pp.378--389, 2009.
    
  \bibitem{DLS14}
    M.~Dom, D.~Lokshtanov, s.~Saurabh.
    \newblock Kernelization Lower Bounds Through Colors and IDs.
    \newblock {\em ACM Trans. Algorithms}, 11(2):13:1--13:20, 2014.     

  \bibitem{DF95}
    R.~G.~Downey, and M.~R.~Fellows.
    \newblock Fixed-Parameter Tractability and Completeness I: Basic Results.
    \newblock {\em SIAM J. Comput.}, 24(4), 873–921, 1995.

  \bibitem{DF95-2}
    R.~G.~Downey, and M.~R.~Fellows.
    \newblock Fixed-parameter tractability and completeness II: On completeness for W[1].
    \newblock {\em Theoretical Computer Science}, 141(1):109–131, 1995.

  \bibitem{DF13}
    R.~G.~Downey, and M.~R.~Fellows.
    \newblock {\em Fundamentals of Parameterized Complexity}.
    \newblock Springer-Verlag, 2013.

  \bibitem{D12}
    A.~Drucker.
    \newblock New Limits to Classical and Quantum Instance Compression.
    \newblock {\em 53rd Annual Symposium on Foundations of Computer Science (IEEE~2012)}, pp.~609--618, 2012.

  \bibitem{FFLRSV09}
    H.~Fernau, F.~V.~Fomin, D.~Lokshtanov, D.~Raible, S.~Saurabh, and Y~Villanger.
    \newblock Kernel(s) for Problems with No Kernel: On Out-Trees with Many Leaves.
    \newblock{\em ACM Trans. Algorithms}, 8:38:1-38:19, 2009.

  \bibitem{FG06}
    J.~Flum, and M.~Grohe.
    \newblock {\em Parameterized Complexity Theory (Texts in Theoretical Computer Science. An EATCS Series)}, Springer-Verlag, 2006.

  \bibitem{FGST13}
    F.~V.~Fomin, S.~Gaspers, S.~Saurabh and S.~Thomass\'{e}.
    \newblock A linear vertex kernel for maximum internal spanning tree.
    \newblock {\em Journal of Computer and System Sciences}, 79(1):1--6, 2013.

  \bibitem{FG04}
    M.~Frick, and M.~Grohe.
    \newblock The complexity of first-order and monadic second-order logic revisited.
    \newblock {\em Annals of Pure and Applied Logic}, 130(1):3--31, 2004.
  
  \bibitem{GJ90}
    M.~R.~Garey, and D.~S.~Johnson.
    \newblock {\em Computers and Intractability; A Guide to the Theory of NP-Completeness}, W.~H.~Freeman \& Co., 1990.

  \bibitem{HK13}
    D.~Hermelin, S.~Kratsch, K.~So{\l}tys, M.~Wahlstr{\"o}m, and X.~Wu.
    \newblock A Completeness Theory for Polynomial (Turing) Kernelization.
    \newblock {\em Parameterized and Exact Computation}, pp.202--215, Springer-Verlag, 2013.

  \bibitem{IPS82}
    A.~Itai, C.~Papadimitriou, and J.~Szwarcfiter.
   \newblock Hamilton Paths in Grid Graphs,
   \newblock {\em SIAM Journal on Computing}, 11(4):676--686, 1982.
   
  \bibitem{LWCC15}
    W.~Li, J.~Wang, J.~Chen, and Y.~Cao.
    \newblock A 2k-vertex Kernel for Maximum Internal Spanning Tree,
    \newblock In {\em Proceedings of Algorithms and Data Structures (WADS 2015)}, LNCS 9214, pp.495--505, Springer-Verlag, 2015.

  \bibitem{LZ18}
    W.~Li and B.~Zhu.
    \newblock A $2k$-kernelization algorithm for vertex cover based on crown decomposition.
    \newblock {\em Theoretical Computer Science}, 739:80--85, 2018.

  \bibitem{NT75}
    G.~L.~Nemhauser, and L.~E.~Trotter.
    \newblock Vertex packings: Structural properties and algorithms.
    \newblock {\em Mathematical Programming}, 8(1):232--248, 1975.

  \bibitem{OY11}
    K.~Ozeki,and T.~Yamashita.
    \newblock Spanning Trees: A Survey.
    \newblock {\em Graphs and Combinatorics}, 27(1):1--26, 2011.
    
    \bibitem{RSV04}
      B.~A.~Reed, K.~Smith, and A.~Vetta:
      \newblock Finding odd cycle transversals. 
      \newblock \emph{Operations Research Letters} 32(4), pp.~299--301, 2004.
      
  \bibitem{STT12}
    H.~Shachnai, G.~Tamir, and T.~Tamir. 
    \newblock A theory and algorithms for combinatorial reoptimization. 
    \newblock In \emph{Proceedings of the 10th Latin American Symposium on
    Theoretical Informatics (LATIN 2012)}, LNCS 7256, Springer, pp.~618--630,
    2012. 

  \bibitem{Sch97}
    M.~W.~Sch{\"a}ffter.
    \newblock Scheduling with forbidden sets.
    \newblock {\em Discrete Applied Mathematics}, 72(1), 155--166, 1997.

\end{thebibliography}

\end{document}